\documentclass[a4paper,12pt]{elsarticle}
\usepackage[T1]{fontenc}
\usepackage[english]{babel}
\usepackage{ae,aecompl}
\usepackage{amsmath, amsthm}
\usepackage{amssymb}
\usepackage[utf8]{inputenc}
\usepackage[section] {placeins}
\usepackage[ruled, vlined, linesnumbered]{algorithm2e}
\newtheorem {Theorem}                 {Theorem}         [section]

\newtheorem {myalgorithm}    [Theorem]  {Algorithm}

\newtheorem {lemma}        [Theorem]  {Lemma}

\input amssym.def
\input amssym

\usepackage{pgfplots}
\usepackage{verbatim}
\usepackage{subfigure}
\usepackage{tikz}
\usetikzlibrary{shapes,arrows,backgrounds,mindmap}
\usepackage{titlesec}
\journal{arXiv}

\address{}

\begin{document} 
\begin{frontmatter}
\title{b-articulation points and b-bridges in strongly biconnected directed graphs}
\author{Raed Jaberi}
\begin{abstract}  

	A directed graph $G=(V,E)$ is called strongly biconnected if $G$ is strongly connected and the underlying graph of $G$ is biconnected. This class of directed graphs was first introduced by Wu and Grumbach. Let $G=(V,E)$ be a strongly biconnected directed graph. An edge $e\in E$ is a b-bridge  if the subgraph $G\setminus \left\lbrace e\right\rbrace =(V,E\setminus \left\lbrace  e\right\rbrace) $ is not strongly biconnected. A vertex $w\in V$ is a b-articulation point if $G\setminus \left\lbrace w\right\rbrace$ is not strongly biconnected, where $G\setminus \left\lbrace w\right\rbrace$ is the subgraph obtained from $G$ by removing $w$.  In this paper we study b-articulation points and b-bridges. 
	
\end{abstract} 
\end{frontmatter}
\section{Introduction}
 A directed graph $G=(V,E)$ is called strongly biconnected if $G$ is strongly connected and the underlying graph of $G$ is biconnected. This class of directed graphs was first introduced by Wu and Grumbach \cite{WG2010}.  Let $G=(V,E)$ be a strongly biconnected directed graph. In this paper we study edges and vertices whose removal destroys strongly biconnectivity in $G$. An edge $e\in E$ is a b-bridge if the subgraph $G\setminus \left\lbrace e\right\rbrace =(V,E\setminus \left\lbrace  e\right\rbrace) $ is not strongly biconnected. A vertex $w\in V$ is a b-articulation point if  $G\setminus \left\lbrace w\right\rbrace$ is not strongly biconnected, where $G\setminus \left\lbrace w\right\rbrace$ is the subgraph obtained from $G$ by removing $w$.  $G$ is $2$-edge-strongly-biconnected (respectively, $2$-vertex-strongly biconnected) if the number of vertices in $G$ is at least $3$ and $G$ has no b-briges (respectively, b-articulation points).
Note that  
each $2$-edge-biconnected directed graph is $2$-edge-connected, but the converse is not necessarily true (see Figure \ref{fig:sbnewconcepts}). Moreover, $2$-vertex-connected directed graphs are not necessarily $2$-vertex-biconnected, as illustrated in figure \ref {fig:figureilustratebap}.

 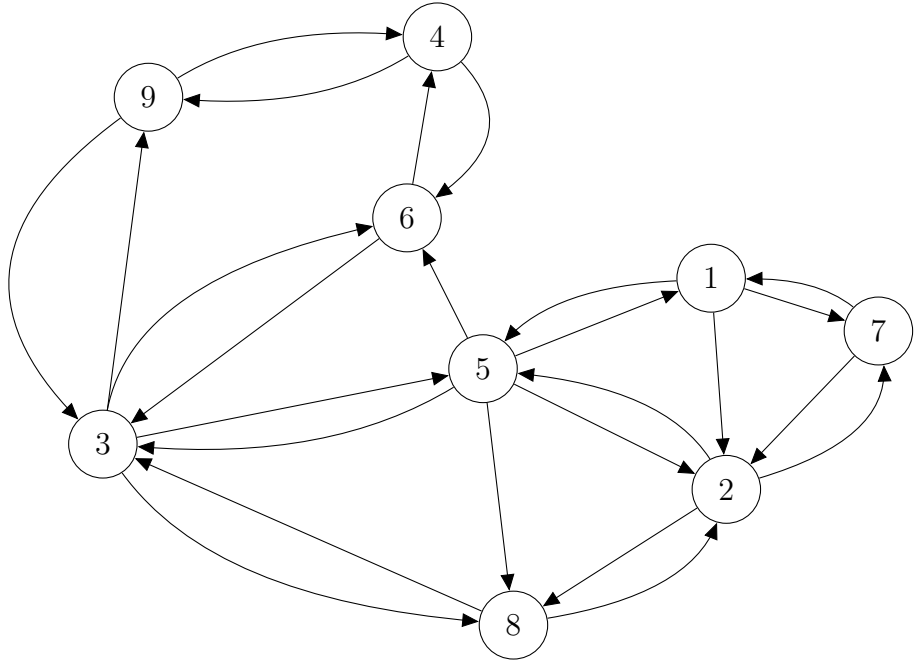
\begin{figure}[htp]
	\centering
	
		\begin{tikzpicture}[xscale=2]
		\tikzstyle{every node}=[color=black,draw,circle,minimum size=0.9cm]
		\node (v1) at  (2.5,0.7){$1$};
		\node (v2) at  (2.6,-2.1){$2$};
		\node (v3) at (-1.5, -1.5) {$3$};
		\node (v4) at (0.7,3.9) {$4$};
		\node (v5) at (1,-0.5) {$5$};
		\node (v6) at (0.5,1.5) {$6$};
		\node (v7) at  (3.6,0){$7$};
		\node (v8) at (1.2,-3.9) {$8$};
		\node (v9) at  (-1.2,3.1){$9$};

		\begin{scope}   
		\tikzstyle{every node}=[auto=right]   
		\draw [-triangle 45] (v9) to [bend left ] (v4);
		\draw [-triangle 45] (v4) to [bend left ](v9);
		\draw [-triangle 45] (v9) to[bend right ] (v3);
		\draw [-triangle 45] (v3) to (v9);
		\draw [-triangle 45] (v4) to [bend left ] (v6);
		\draw [-triangle 45] (v6) to (v4);
		\draw [-triangle 45] (v3) to (v5);
		\draw [-triangle 45] (v5) to (v8);
		\draw [-triangle 45] (v8) to (v3);
		\draw [-triangle 45] (v5) to (v2);
		\draw [-triangle 45] (v2) to (v8);
	    \draw [-triangle 45] (v5) to (v6);
	    \draw [-triangle 45] (v3) to[bend right ] (v8);
	    	\draw [-triangle 45] (v2) to [bend right ] (v5);
	    	\draw [-triangle 45] (v3) to [bend left ] (v6);
	    \draw [-triangle 45] (v6) to (v3);
	    \draw [-triangle 45] (v8) to [bend right ] (v2);
	    \draw [-triangle 45] (v5) to [bend left ] (v3);
	    \draw [-triangle 45] (v5) to (v1);
	    \draw [-triangle 45] (v1) to (v7);
	    \draw [-triangle 45] (v7) to (v2);
	    \draw [-triangle 45] (v1) to (v2);
	    \draw [-triangle 45] (v1) to [bend right ] (v5);
	    \draw [-triangle 45] (v7) to [bend right ] (v1);
	    \draw [-triangle 45] (v2) to [bend right ] (v7);
		\end{scope}
		\end{tikzpicture}
	\caption{A strongly biconnected graph $G$. This graph is $2$-edge-connected since it has no strong birdges. But $G$ contains a b-bridge $(5,6)$. Thus, $G$ is not $2$-edge-strongly biconnected.}
	\label{fig:sbnewconcepts}
\end{figure}
\begin{figure}[htp]
	\centering
	\scalebox{0.96}{
		\begin{tikzpicture}[xscale=2]
		\tikzstyle{every node}=[color=black,draw,circle,minimum size=0.9cm]
		\node (v1) at (2.2,3.1) {$1$};
		\node (v2) at (-2.5,0) {$2$};
		\node (v5) at (-0.5, -2.5) {$5$};
		\node (v4) at (3.6,-1){$4$};
		\node (v3) at (1,0)  {$3$};
		\node (v6) at  (-1.2,3.1){$6$};

		\begin{scope}   
		\tikzstyle{every node}=[auto=right]   
	  	\draw [-triangle 45] (v1) to (v6);
	  	\draw [-triangle 45] (v5) to (v6);
	  \draw [-triangle 45] (v6) to [bend left ](v1);
	  \draw [-triangle 45] (v1) to (v4);
	  \draw [-triangle 45] (v4) to [bend right ](v1);
	  \draw [-triangle 45] (v3) to [bend right ] (v4);
	  \draw [-triangle 45] (v4) to (v3);
	  \draw [-triangle 45] (v3) to [bend left ](v5);
	  \draw [-triangle 45] (v2) to (v5);
	  \draw [-triangle 45] (v5) to (v3);
	  \draw [-triangle 45] (v5) to [bend left ](v2);
	  \draw [-triangle 45] (v2) to (v6);
	  \draw [-triangle 45] (v6) to [bend right ](v2);
	  \draw [-triangle 45] (v1) to (v3);
	  \draw [-triangle 45] (v2) to (v6);
	  \draw [-triangle 45] (v2) to (v3);
	  \draw [-triangle 45] (v6) to (v3);
	  \draw [-triangle 45] (v4) to (v6);
	
		\end{scope}
		\end{tikzpicture}}
	\caption{A strongly biconnected graph $G=(V,E)$.  $G$ is $2$-vertex-connected. Since vertex $6$ is a $b$-articulation point, $G$ is not $2$-vertex-strongly biconnnected. }
	\label{fig:figureilustratebap}
\end{figure}
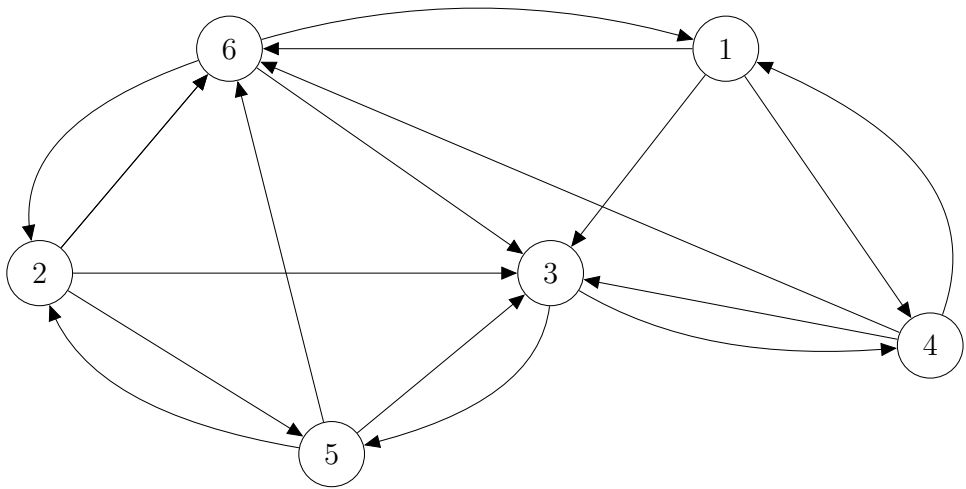

Articulation points and bridges of an undirected graph can be identified in linear time \cite{T72,T74,JS13}. In $2010$, Georgiadis \cite{G10} gave a linear time algorithm to test whether a directed graph is $2$-vertex-connected. Italiano et al. \cite{ILS12} gave linear time algorithms for calculating  strong articulation points and strong bridges in directed graphs. Wu and Grumbach \cite{WG2010} introduced a class of directed graphs, called strongly biconnected graphs.
In this paper we study  b-articulation points and b-bridges.

\section{Computing b-bridges}
This section illustrates how to compute b-bridges in strongly biconnected graphs. 
\begin{lemma}\label{def:sbbridgelemma}
	Let $G=(V,E)$ be a strongly biconnected directed graph and let $e$ be a strong bridge in $G$. Then $e$ is a b-bridge.
\end{lemma}
\begin{proof}
	Since $e$ is a strong bridge in $G$ and $G$ is strongly connected, the subgraph $(V,E\setminus\left\lbrace e\right\rbrace )$ is not strongly connected. Therefore, $e$ is a b-bridge.
\end{proof}
Note that b-bridges are not necessarily strong bridges, as shown in Figure \ref{fig:sbnewconceptssbbr}.

\begin{figure}[h!]
	\centering
	
	\begin{tikzpicture}[xscale=2]
	\tikzstyle{every node}=[color=black,draw,circle,minimum size=0.9cm]
	\node (v1) at  (2.4,0.7){$1$};
	\node (v2) at  (2.5,-2.1){$2$};
	\node (v3) at (-1.4, -1.5) {$3$};
	\node (v4) at (0.7,3.9) {$4$};
	\node (v5) at (1,-0.5) {$5$};
	\node (v6) at (0.5,1.5) {$6$};
	\node (v7) at  (3.6,0){$7$};
	\node (v8) at (1.2,-3.6) {$8$};
	\node (v9) at  (-1.2,3.1){$9$};
	\node (v10) at  (3,3){$10$};	
	\begin{scope}   
	\tikzstyle{every node}=[auto=right]   
	\draw [-triangle 45] (v9) to [bend left ] (v4);
	\draw [-triangle 45] (v4) to [bend left ](v9);
	\draw [-triangle 45] (v9) to[bend right ] (v3);
	\draw [-triangle 45] (v3) to (v9);
		\draw [-triangle 45] (v10) to (v6);

	\draw [-triangle 45] (v4) to [bend left ] (v10);
	\draw [-triangle 45] (v6) to (v4);
	\draw [-triangle 45] (v3) to (v5);
	\draw [-triangle 45] (v5) to (v8);
	\draw [-triangle 45] (v8) to (v3);
	\draw [-triangle 45] (v5) to (v2);
	\draw [-triangle 45] (v2) to (v8);
	\draw [-triangle 45] (v5) to (v6);
	\draw [-triangle 45] (v3) to[bend right ] (v8);
	\draw [-triangle 45] (v2) to [bend right ] (v5);
	\draw [-triangle 45] (v6) to (v3);
	\draw [-triangle 45] (v8) to [bend right ] (v2);
	\draw [-triangle 45] (v5) to [bend left ] (v3);
	\draw [-triangle 45] (v5) to (v1);
	\draw [-triangle 45] (v1) to (v2);
	\draw [-triangle 45] (v1) to [bend right ] (v5);
	\draw [-triangle 45] (v7) to [bend right ] (v1);
	\draw [-triangle 45] (v2) to [bend right ] (v7);
	\end{scope}
	\end{tikzpicture}
	\caption{A strongly biconnected graph $G$. Edge $(2,7)$ is both a strong bridge and a b-bridge. Note that the underlying graph of $G\setminus\left\lbrace (5,6) \right\rbrace $ has an articulation point. Thus, $(5,6)$ is a b-bridge. If we remove $(5,6)$ from $G$, the reamaining subgraph is still strongly connected. Therefore, $(5,6)$ does not form a strong bridge in $G$} 
	\label{fig:sbnewconceptssbbr}
\end{figure}
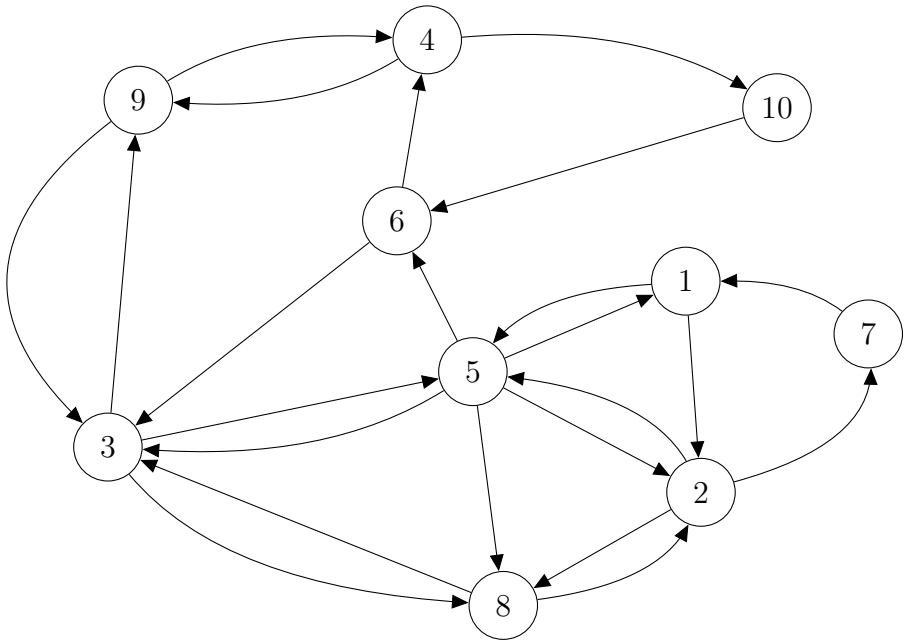

Wu and Grumbach \cite{WG2010} introduced strongly biconnected components. A strongly biconnected component of a strongly connected graph $G=(V,E)$ is a maximal vertex subset $U\subseteq V$ such that the induced subgraph on $U$ is strongly biconnected.	

Let $G$ be a strongly biconnected directed graph. The following lemma shows how to decrease the number of strongly biconnected components in a strongly connected subgraph of $G$.
\begin{lemma} \label{def:lemmaaddingedge}
	Let $G_s=(V,E_s)$ be a subgraph of a strongly biconnected directed graph $G=(V,E)$ such that $G_s$ is strongly connected and $G_s$ has $t$ strongly biconnected components. Let $(u,w)$ be an edge in $E\setminus E_s$  such that $u,w $ are in distinct strongly
	biconnected components $C_{1}^{B},C_{2}^{B}$ of $G_{s}$ with $u,w \notin C_{1}^{B} \cap C_{2}^{B}$. Then the graph $(V,E\cup \left\lbrace (u,w) \right\rbrace )$ contains at most $t-1$ strongly biconnected components.
\end{lemma}
\begin{proof}
 Since $G_s$ is strongly connected, there exists a simple path $P$ from $w$ to $u$ in $G_s$. Path $p$ and edge $(u,w)$ form a simple cycle. Consequently, $u,w$ are in the same strongly biconnected component of $(V,E\cup \left\lbrace (u,w) \right\rbrace )$.
\end{proof}

Algorithm \ref{algo:algorithmforbridges} can compute all the b-bridges of a strongly biconnected directed graph.
\begin{figure}[h]
	\begin{myalgorithm}\label{algo:algorithmforbridges}\rm\quad\\[-5ex]
		\begin{tabbing}
			\quad\quad\=\quad\=\quad\=\quad\=\quad\=\quad\=\quad\=\quad\=\quad\=\kill
			\textbf{Input:} A strongly biconnected graph $G=(V,E)$.\\
			\textbf{Output:} $B$, where $B$ is the set of  b-bridges in $G$\\
			{\small 1}\> $U \leftarrow \emptyset$\\
			{\small 2}\> Calculate the set of strong bridges in $G$\\
			{\small 3}\>\textbf{for} every strong bridge $(v,w)$ in $G$ \textbf{do} \\
			{\small 4}\>\>$U \leftarrow U\cup \left\lbrace (v,w)\right\rbrace $ \\
			{\small 5}\> select a vertex $y \in V $\\
				{\small 6}\> $E_{y} \leftarrow \emptyset$\\
			{\small 7}\> build a spanning tree $T$ rooted at $y$ in $G$\\
			{\small 8}\>\textbf{for} each edge $(v,w)$ in $T$ \textbf{do} \\
			{\small 9}\>\>$E_{y} \leftarrow E_{y}\cup \left\lbrace (v,w)\right\rbrace $ \\
			{\small 10}\> build $G_{r}=(V,E_{r}),$ where $E_{r}=\left\lbrace (v,w)  \mid (w,v)\in E  \right\rbrace $\\
			{\small 11}\> create a spanning tree $T_{y}$ rooted at $y$ in $G_{r}$\\
			{\small 12}\>\textbf{for} each edge $(v,w)$ in $T_y$ \textbf{do} \\
			{\small 13}\>\>$E_{y} \leftarrow E_{y}\cup \left\lbrace (w,v)\right\rbrace $ \\
			{\small 14}\> \textbf{while} $G_{y}=(V,E_{y})$ is not  strongly biconnected \textbf{do}\\ 
			{\small 15}\>\> identify the strongly biconnected components of $G_{y}$\\
			{\small 16}\>\> find an edge $(u,w) \in E\setminus E_{y}$ such that $u,w $ are in distinct strongly\\
			{\small 17}\>\>\>  biconnected components $C_{1}^{B},C_{2}^{B}$ of $G_{y}$ with $u,w \notin C_{1}^{B} \cap C_{2}^{B}$\\
			{\small 18}\>\>  $E_{y}\leftarrow E_{y} \cup \left\lbrace (u,w) \right\rbrace  $.\\
			{\small 19}\> $B \leftarrow U$.\\ 
			{\small 20}\> \textbf{for} every edge $(u,w)\in E_y \setminus B$ \textbf{do} \\
			{\small 21}\>\> \textbf{if}  $G \setminus \left\lbrace  (u,w)\right\rbrace $ \ is not strongly biconnected \textbf{then}\\
			{\small 22}\>\>\> $B\leftarrow B \cup \left\lbrace e\right\rbrace  $.
		\end{tabbing}
	\end{myalgorithm}
\end{figure}
\begin{lemma}\label{def:edgesnotbbridgess}
	Let $e\in E\setminus E_y$. Then $e$ is not a b-bridge.
\end{lemma}
\begin{proof}
	It is known that a strongly connected subgraph of a strongly connected graph can be obtained by finding a spanning tree in $G$ and in $G_r$ (see \cite{FJ81}). By lemma \ref{def:lemmaaddingedge}, the number of strongly biconnected components must decrease by at least one in each iteration of while loop. Therefore, the subgraph $(V,E_y)$ is strongly biconnected.
\end{proof}
The correctness of Algorithm \ref{algo:algorithmforbridges} follows from Lemma \ref{def:sbbridgelemma}, Lemma \ref{def:lemmaaddingedge}, and Lemma \ref{def:edgesnotbbridgess}.

\begin{Theorem}
	The Algorithm \ref{algo:algorithmforbridges} runs in time $O(nm)$.
\end{Theorem}
\begin{proof}
	 The set of strong bridges $U$ can be calculated in linear time using the algorithm of Italiano et al. \cite{ILS12}.
		A spanning tree of a strongly biconnected graph can be constructed in linear time using depth first search. Thus, lines $7$--$13$ take linear time.  By Lemma \ref{def:lemmaaddingedge}, the number of iterations of the while loop is at most $n-1$. Moreover, the strongly biconnected components of a strongly connected graph can be found in linear time \cite{WG2010}. Therefore, the while loop of lines $14$--$18$ takes time $O(nm)$. The for loop of lines $20$--$22$ executes $n-|U|$ times. Thus, the total time taken by this loop is $O(nm)$.
	  
\end{proof}

\section{Open problem}
It is easy to see that the b-articulation points of a strongly biconnected directed graph $G$ can be computed in time $O((n-a)m)$, where $a$ is the number of strong articulation points in $G$. Strong articulation points can be computed in linear time \cite{ILS12,FILOS12}.
We leave as an open problem  whether there is a linear time algorithm for calculating b-bridges and b-articulation points in directed graphs.
. 

\end{document}